\documentclass[smallextended]{svjour3}       
\smartqed  
\usepackage{graphicx}
\usepackage{latexsym}

\newcommand{\cA}{{\mathcal A}}
\newcommand{\cQ}{{\mathcal Q}}
\newcommand{\cF}{{\mathcal F}}
\newcommand{\cR}{{\mathcal R}}

\journalname{}
\begin{document}

\title{Self-sustaining autocatalytic networks within open-ended  reaction systems
}
\subtitle{}


\author{Mike Steel}


\institute{M. Steel \at
              Biomathematics Research Centre, University of Canterbury, Christchurch, New Zealand \\
              Tel.: +64-3-364-2987 ext 7688\\
              Fax: +64-3-364-2587\\
              \email{mike.steel@canterbury.ac.nz}           
}

\date{Received: date / Accepted: date}

\maketitle

\begin{abstract}
Given any finite and closed chemical reaction system, it is possible to efficiently determine whether or not it contains a `self-sustaining and collectively autocatalytic' subset of reactions, and to find such subsets when they exist.  However, for systems that are potentially open-ended (for example, when no prescribed upper bound is placed on the complexity or size/length of molecules types), the theory developed for the finite case breaks down.  We investigate a number of subtleties that arise in such systems that are absent in the finite setting, and present  several new results. 
\end{abstract}

\section{Introduction}
\label{intro}
Consider any system of chemical reactions, in which certain molecule types catalyse reactions and where there is a pool of  simple molecule types available from the environment (a `food source').  One can then ask whether, within this system, there is a  subset of reactions that is both self-sustaining (each molecule can be constructed starting just from the food source) and collectively autocatalytic (every reaction is catalysed by some molecule produced by the system or present in the food set) \cite{kau1},\cite{kau2}.   This notion of `self-sustaining and collectively autocatalytic'  needs to be carefully formalised (we do so below), and is relevant to some basic questions such as how biochemical metabolism  began at the origin of life \cite{hig}, \cite{sm}, \cite{vas}.  A simple mathematical framework for formalising and studying such self-sustaining autocatalytic networks has been developed -- so-called `RAF (Reflexively-autocatalytic and F-generated) theory'. This theory includes an algorithm to determine whether such networks exists within a larger system, and for classifying these networks; moreover, the theory allows us to calculate the probability of the formation of such systems within networks based on the ligation and cleavage of polymers, and a random pattern of catalysis.

However, this theory relies heavily on the system being closed and finite. In certain settings, it is useful to consider polymers of arbitrary length being formed (e.g. in generating the membrane for a protocell \cite{horn}). In these and other unbounded chemical systems, interesting  complications arise for RAF theory,  particularly where the catalysis of certain reactions is possible only by molecule types that are  of greater complexity/length than the reactants or product of the reactions in question.  In this paper, we extend earlier RAF theory to  deal with unbounded chemical reaction systems. As in some of our earlier work, our analysis ignores the dynamical aspects, which are dealt with in other frameworks, such as `chemical organisation theory' \cite{dit}; here we concentrate instead on just the pattern of catalysis and the availability of reactants.

\subsection{Preliminaries and definitions}

In this paper, a {\em chemical reaction system} (CRS) consists of (i) a set $X$ of molecule types, (ii) a set $\cR$ of reactions, (iii) a pattern of catalysis $C$ that describes which molecule(s) catalyses which reactions, and (iv) a distinguished subset $F$ of $X$ called the {\em food set}.

We will denote a CRS as a quadruple  $\cQ = (X,\cR, C, F)$, and encode the pattern of catalysis $C$ by specifying a subset of $X \times \cR$ so that $(x,r) \in C$ precisely if
molecule type $x$ catalyses reaction $r$. See Fig.~\ref{fig1} for a simple example (from \cite{sm}).

\begin{figure}[htb]
\centering
\includegraphics[scale=1.0]{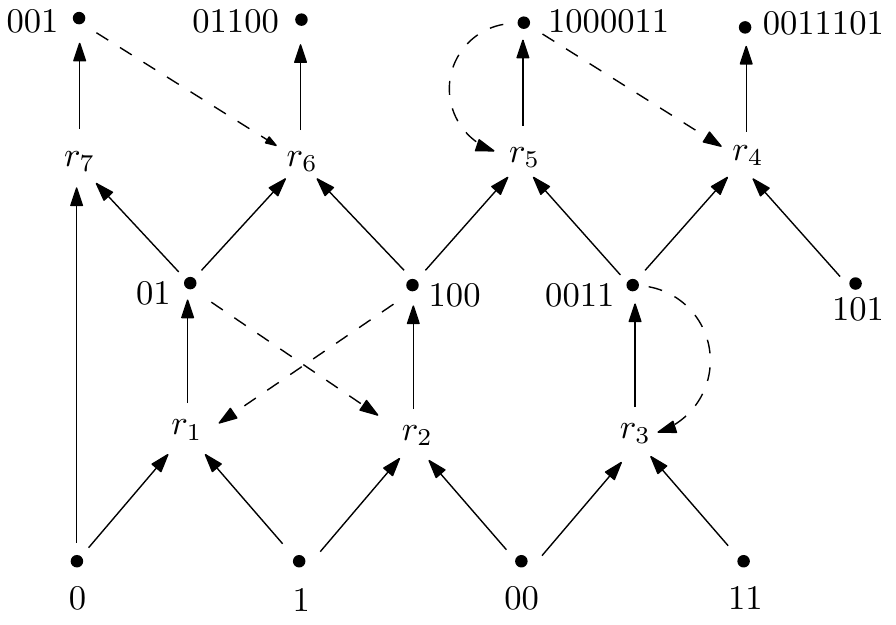}
\caption{A simple CRS based on polymers over a two-letter alphabet (0,1), with a food set $F=\{0,1, 00, 11\}$ and seven reactions. Dashed arrows indicate catalysis; solid arrows show reactants entering a reaction and products leaving. In this CRS there are exactly four RAFs (defined below), namely $\{r_1, r_2\}$, $\{r_3\}$, $\{r_1, r_2, r_3\}$, and $\{r_1, r_2, r_3, r_5\}$.  }
\label{fig1}
\end{figure}

In certain applications, $X$ often consist of -- or at least contain --  a set of polymers (sequences) over some finite alphabet $\cA$ (i.e. chains $x_1x_2 \cdots x_r$, $r\geq 1$, where $x_i \in \cA$), as in Fig.~\ref{fig1}; such polymer systems are particularly relevant to RNA or amino-acid sequence models of early life. Reactions involving such polymers  typically involve cleavage and ligation (i.e. cutting and/or joining polymers), or adding or deleting a letter to an existing chain.   Notice that if no bound is put on the maximal length of the polymers, then both $X$ and $\cR$ are infinite for such networks, even when $|\cA|=1$.  

In this paper we do not necessarily assume that $X$ consists of polymers, or that the reactions are of any particular type. Thus, a reaction can be viewed formally as an ordered pair $(A,B)$ consisting of a multi-set $A$ of elements from $X$ (the reactants of $r$) and a multi-set $B$ of elements of $X$ (the products of $r$); but we will mostly use the equivalent and more conventional notation of writing a reaction in the form:
$$r= (a_1+a_2 +\cdots+a_k \rightarrow b_1+b_2+\cdots+b_l),$$ where the $a_i$'s (reactants of $r$) and $b_j$'s (products of $r$) are elements of $X$, and $k, l\geq 1$ (e.g.  $x \rightarrow y, x+ x \rightarrow y$ and $y \rightarrow x+x'$ are reactions).

In this paper, we extend our earlier analysis of RAFs to
the general (finite or infinite) case and find that certain subtleties arise that are absent in the finite case.
We  will mostly assume the following conditions  (A1) and (A2), and sometimes also (A3).
\begin{itemize}
\item[(A1)] $F$ is finite; 
\item[(A2)]  each reaction $r \in \cR$ has a finite set of reactants, denoted $\rho(r)$, and a finite set of products, denoted $\pi(r)$;
\item[(A3)]  for any given finite set $Y$ of molecule types, there are only finitely many reactions $r$ with $\rho(r)=Y$.
\end{itemize}
Given a subset $\cR'$ of $\cR$, we say that a subset $W \subseteq X$ of molecule types is {\em closed} relative to $\cR'$ if $W$ satisfies the property
$r \in \cR' \mbox{ and } \rho(r) \subseteq W \Rightarrow \pi(r) \subseteq W.$
In other words, a set of molecule types is closed relative to $\cR'$ if every molecule that can be produced from $W$ using reactions in $\cR'$ is already present in $W$. 
Notice that the full set $X$ is itself closed. The {\em global closure} of $F$ relative to $\cR'$, denoted here as ${\rm gcl}_{\cR'}(F)$, is the intersection of all closed sets that contain $F$ (since $X$ is closed, this intersection is well defined). 
Thus ${\rm gcl}_{\cR'}(F)$  is the unique minimal set of molecule types  containing $F$ that is closed relative to $\cR'$.

We can also consider a {\em constructive closure} of $F$ relative to $\cR'$, denoted here as ${\rm ccl}_{\cR'}(F)$, which is union of the  set $F$ and the set of molecule types $x$ that can be obtained from $F$ by carrying out any finite sequence of reactions from $\cR'$ where, for each reaction $r$ in the sequence, each reactant of $r$ is either an elements of $F$ or a product of a reaction occurring earlier in the sequence, and $x$ is a product of the last reaction in the sequence.

Note that ${\rm gcl}_{\cR'}(F)$ always contains ${\rm ccl}_{\cR'}(F)$ (and these two sets coincide when the CRS is finite) but,  for an infinite CRS, ${\rm ccl}_{\cR'}(F)$ can be a strict subset of ${\rm gcl}_{\cR'}(F)$, even when (A1) holds.
To see this, consider the system $(X, \cR)$ where $X=\{x_0, x_1, x_2, \ldots\}$, $F=\{f\}$, where $\cR'=\{r_0, r_1, r_2, r_3, \ldots\}$ is defined as follows:
$$r_1 = (f \rightarrow x_1);$$
$$ r_j = (f+ x_j \rightarrow x_{j+1}), \mbox{ for all }  j \geq 1;$$
$$r_0 = (x_1+x_2+ \cdots \rightarrow x_0).$$
Then $x_0 \in {\rm gcl}_{\cR'}(F)-{\rm ccl}_{\cR'}(F)$.
In this example, notice that $r_0$ has infinitely many reactants, which violates (A2).  By contrast, when (A2) holds, we have the following result.

\begin{lemma}
\label{lem1}
Suppose that  (A2) holds. Then  ${\rm ccl}_{\cR'}(F) ={\rm gcl}_{\cR'}(F)$.
Moreover, under (A1) and (A2), if $\cR'$ is countable, then this (common) closure of $F$ relative to $\cR'$ is countable also.

\end{lemma}
\begin{proof}
Suppose the condition of Lemma~\ref{lem1} holds but that  ${\rm ccl}_{\cR'}(F)$ is not closed; we will derive a contradiction. Lack of closure means there is a molecule
$x$ in $X- {\rm ccl}_{\cR'}(F)$ which is the product of some reaction $r \in \cR'$ that has all its reactants in ${\rm ccl}_{\cR'}(F)$. By (A2), the set of reactants
of $r$ is finite, so we may list them as $x_1, x_2, \ldots, x_k$, and, by the definition of ${\rm ccl}_{\cR'}(F)$, for each $i \in \{1, \ldots, k\}$,  either $x_i \in F$ or there is a  finite sequence $S_i$ 
 of reactions from $\cR'$ that generates $x_i$ starting from reactants entirely in $F$ and using just elements of $F$ or products of reactions appearing
 earlier in the sequence $S_i$.  By concatenating these sequences (in any order) and appending $r$ at the end, we obtain a finite sequence of reactions that generate $x$ from $F$,
which contradicts the assumption that ${\rm ccl}_{\cR'}(F)$ is not closed. If follows that ${\rm ccl}_{\cR'}(F)$ is closed relative to $\cR'$, and since it is clearly a minimal set
containing $F$ that is closed relative to $\cR'$, it follows that ${\rm ccl}_{\cR'}(F) =  {\rm gcl}_{\cR'}(F)$.
That ${\rm ccl}_{\cR'}(F)$ is countable under (A1) and (A2) follows from the fact that any countable union of finite sets is countable.
\hfill$\Box$
\end{proof}
 In view of Lemma~\ref{lem1}, whenever (A2) holds, we will henceforth denote the (common) closure of $F$ relative to $\cR'$ as ${\rm cl}_{\cR'}(F)$.

\bigskip

\noindent {\bf Definition [RAF, and related concepts]}
 Suppose we have a CRS $\cQ=(X, \cR, C, F)$, which satisfies condition (A2).   An RAF  for $\cQ$ is a non-empty subset $\cR'$ of $\cR$ for which 
\begin{itemize} 
\item[(i)] for each $r\in \cR'$,  $\rho(r) \subseteq {\rm cl}_{\cR'}(F)$; and
\item[(ii)] for each $r\in \cR'$, at least one molecule type in ${\rm cl}_{\cR'}(F)$ catalyses $r$. 
\end{itemize}

In words, a non-empty set $\cR'$ of reactions forms  an RAF for $\cQ$ if, for every reaction $r$ in $\cR'$, each reactant of $r$ and at least one catalyst of $r$ is either present in $F$ or able to be constructed from $F$ by using just reactions from within the set $\cR'$. 
\bigskip

\noindent An RAF $\cR'$ for $\cQ$ is said to be a {\em finite RAF} or an {\em infinite RAF} depending on whether or not $|\cR'|$ is finite or infinite. 
The concept of an RAF is a formalisation of a `collectively autocatalytic set', pioneered by Stuart Kauffman \cite{kau1} and \cite{kau2}.  
Since the union of any collection of RAFs is also an RAF, any CRS that contains an RAF  necessarily contains a unique maximal RAF. 
An {\em irrRAF} is an (infinite or finite) RAF that is minimal -- i.e. it contains  no RAF as a strict subset.  In contrast to the uniqueness of the maximal RAF, a finite CRS can have exponentially many irrRAFs \cite{hor4}.

The RAF concept
needs to be distinguished from the stronger notion of a  {\em constructively autocatalytic and F-generated} (CAF) set \cite{mos} which requires that $\cR'$ can be ordered
$r_1, r_2, \ldots, r_N$ so that all the reactants and at least one catalyst of $r_i$ are present in ${\rm cl}_{\{r_1, \ldots, r_{i-1}\}}(F)$ for all $i\in \{1, \ldots, N\}$ (in the initial case where $i=1$, we take ${\rm cl}_\emptyset(F) = F$).  
This condition essentially means that in a CAF, a reaction can only proceed if one of its catalysts is  already available, whereas an RAF could become established by allowing one or more reactions  $r$ to proceed uncatalysed (presumably at a much slower rate) so that later, in some chain of reactions, a catalyst for $r$ is generated, allowing the whole system to `speed up'.   Notice
that although the CRS in Fig.~\ref{fig1} has four RAFs it has no CAF.

\begin{figure}[htb]
\centering
\includegraphics[scale=0.5]{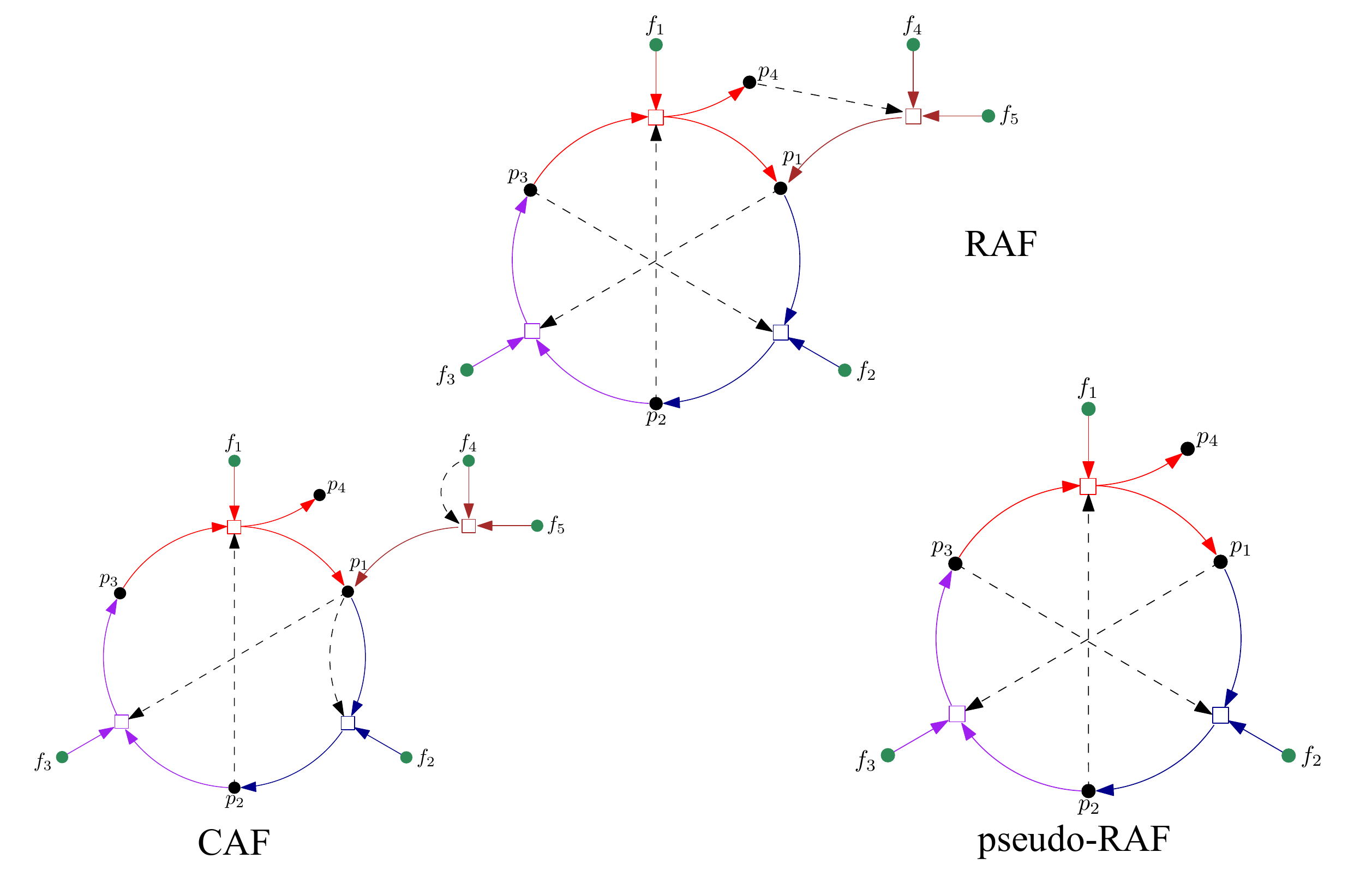}
\caption{Examples of a  finite RAF (that is not a CAF), a finite CAF and a finite pseudo-RAF (that is not an RAF).  In these examples, the molecule types are round nodes (the food set is denoted $f_1, f_2, \ldots$, and $p_1, p_2, \ldots$ are products), reactions are hollow squares, and dashed arrows indicate catalysis. }
\label{fig1b}
\end{figure}

The RAF concept also needs to be distinguished from the weaker notion of a  {\em pseudo-RAF} \cite{ste}, which replaces condition (ii) with the relaxed condition:
\begin{quote}
(ii)$'$: for all $r\in \cR'$, there exists $x\in F$ or $x\in \pi(r)$ for some $r\in \cR'$  such that $(x,r) \in C$.  
\end{quote}
\noindent In other words,  a pseudo-RAF that fails to be an RAF is an autocatalytic system that could continue to persist once it exists, but it can never form from just the food set $F$, since it is not $F$-generated. 

These two  alternatives notions to RAFs are illustrated (in the finite setting) in  Fig.~\ref{fig1b}.  
Notice that every CAF is an RAF and every RAF is a pseudo-RAF, but these containments are strict, as Fig.~\ref{fig1b} shows.

While the notion of a CAF may seem reasonable, it is arguably too conservative in comparison to an RAF, since a reaction can still proceed if no catalyst is present, albeit it at a much slower rate, allowing the required catalyst to eventually be produced.
However relaxing the RAF definition further to a pseudo-RAF is problematic (since a reaction cannot proceed at all, unless all its reactants are present, and so such a system cannot arise spontaneously just from $F$). This, along with other desirable properties of RAFs (their formation requires only low levels of catalysis in contrast to CAFs \cite{mos}), suggests that RAFs are a reasonable candidate for capturing the minimal necessary condition for self-sustaining autocatalysis, particularly in models of the origin of metabolism.

\subsection{Properties of RAFs in an infinite CRS}

As in the finite CRS setting, the union of all RAFs is an RAF, so any CRS that contains an RAF has a unique maximal one.  It is easily seen that an infinite CRS that contains an RAF need not have a maximal finite RAF, even under (A1)--(A3), but in this case, the CRS would necessarily also contain an infinite RAF (the union of all the finite RAFs). 

A natural question is the following:  if an infinite CRS contains an infinite RAF, does it also contain a finite one? 
It is easily seen that even under conditions (A1) and (A2), the answer to this last question is `no'.  We provide three examples to illustrate different ways in which  this can occur.   This is in contrast to CAFs, for which exactly the opposite holds:  if a CRS contains an infinite CAF, then it necessarily contains a sequence of finite ones.  Moreover, two of the infinite RAFs in the following example contain no irrRAFs (in contrast to the finite case, where every RAF contains at least one irrRAF). 

\bigskip

\noindent {\bf Example 1:}    Let $X =\{f, x_1, \ldots, x_n, \ldots\}$,  $F= \{f\}$  and $\cR= \{r_1, r_2, \ldots, r_n, \ldots\}$.
Let  $r_1= ( f \rightarrow x_1)$. We will specify particular CRS's  by describing $r_2, r_3, \ldots$, and the pattern of catalysis as follows.

\begin{itemize}
\item  $\cQ_1$ has a reaction $r_i = (f+x_{i-1} \rightarrow x_i)$ for each $i>1$ and $r_i$ is catalysed by $x_{i+1}$ for each $i\geq 1$.  
\item $\cQ_2$ has a reaction $r_i = (f+f+ \cdots + f  [i \mbox{ times}] \rightarrow x_i)$ for each $i>1$ and $r_i$ is catalysed by  $x_{i+1}$  for each $i\geq 1$. 
\item $\cQ_3$ has the same reactions as $\cQ_2$ but $r_i$ is now catalysed by every $x_j: j>i$.
\end{itemize}

Fig.~\ref{figx} illustrates the three CRS's. 

\begin{figure}[htb]
\centering
\includegraphics[scale=0.8]{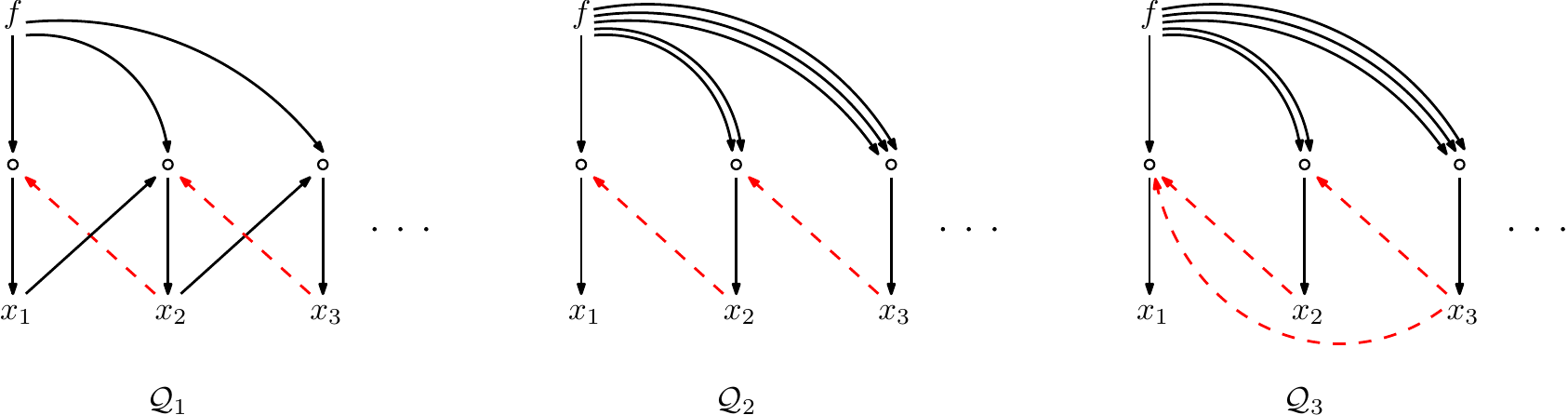}
\caption{Three simple examples of  infinite CRSs that have infinite RAFs but no finite RAF.}
\label{figx}
\end{figure}

Each of $\cQ_1, \cQ_2, \cQ_3$ satisfy (A1) and (A2), but only $\cQ_1$ satisfies (A3). All three CRSs  contain infinite RAFs, but no finite RAF, and no CAF.   More precisely: 
\begin{itemize}
\item
$\cQ_1$ has $\cR$ as its unique RAF (which is therefore an irrRAF).
\item
The RAFs of $\cQ_2$ consist precisely of all subsets of $\{r_j, r_{j+1}, \ldots, \}$ for some $j$.
Thus $\cQ_2$ has a countably infinite number of RAFs but no irrRAF.
\item
The RAFs of $\cQ_3$ consist precisely  of all  infinite subsets of  $\cR$.
Thus, the set of RAFs for $\cQ_3$ in uncountably infinite, and it contains no irrRAF.
\end{itemize}

\section{Determining whether or not a CRS contains an RAF}
\label{determ}
In this section, we assume that both (A1) and (A2) hold.
Given a CRS $\cQ=(X, \cR, C, F)$, consider the following nested decreasing sequence of reactions: 
$\cR_1, \cR_2, \ldots,$ defined by $\cR_1 = \cR$ and for each $i>1$:
\begin{equation}
\label{R1eqx} 
\cR_{i+1} = \{r \in \cR_i: \rho(r) \subseteq {\rm cl}_{\cR_i}(F), \mbox{ and } \exists x \in {\rm cl}_{\cR_i}(F): (x,r) \in C\}.
\end{equation}
Thus, $\cR_{i+1}$ is obtained from $\cR_i$ by removing any reaction that fails to have either all its reactants or at least one catalyst in the closure of $F$ relative to $\cR_i$.  
Let $\mu(\cQ) = \bigcap_{i\geq 1} \cR_i$.
It is easily shown that any RAF $\cR'$ present in $\cQ$ is necessarily a subset of $\mu(\cQ)$ (since $\cR' \subseteq \cR_i$ for all $i\geq 1$ by induction on $i$). Thus if $\mu(\cQ) = \emptyset$ then
$\cQ$ does not have an RAF.  In the finite case there is a strong converse -- if $\mu(\cQ) \neq \emptyset$ then
$\cQ$ has an RAF, and $\mu(\cQ)$ is the unique maximal RAF for $\cQ$ (this is the basis for the `RAF algorithm' \cite{hor2} and \cite{hor3}).   However, in contrast, this result can fail for an infinite CRS, as we now show with a simple example, which also satisfies (A1)--(A3). 

\bigskip

\noindent {\bf Example 2: }  
Consider  the following infinite CRS, $\cQ_4= (X, \cR, C, F)$, where $F= \{f\}$,  and
$X=\{f, s,t,\} \cup \{x_1, x_2, x_3,  \ldots\} \cup \cF,$
where $$\cF = \{f, ff, fff, \cdots, f^{(i)}, \cdots\}$$ (this set can be thought of as all polymers of $f$).
The reaction set is 
$\cR= \{r_1,r_2, r_3, \ldots\} \cup \{ r'_2, r'_3,  \ldots\}$,
where, for all $i \geq 1:$
$$r_i = (f+ f^{(i)} \rightarrow f^{(i+1)}); $$
$$r'_i = (f^{(i)} \rightarrow x_i +s).$$
The pattern of catalysis is defined as follows:
$s$ catalyses $r_1$ and $t$ catalyses $r'_2$, and for all $i>1$
$f^{(i)}$ catalyses $r_i$ and $x_i$ catalyses $r'_{i+1}$.
This CRS is illustrated in Fig~\ref{figy}.
\begin{figure}[htb]
\centering
\includegraphics[scale=0.8]{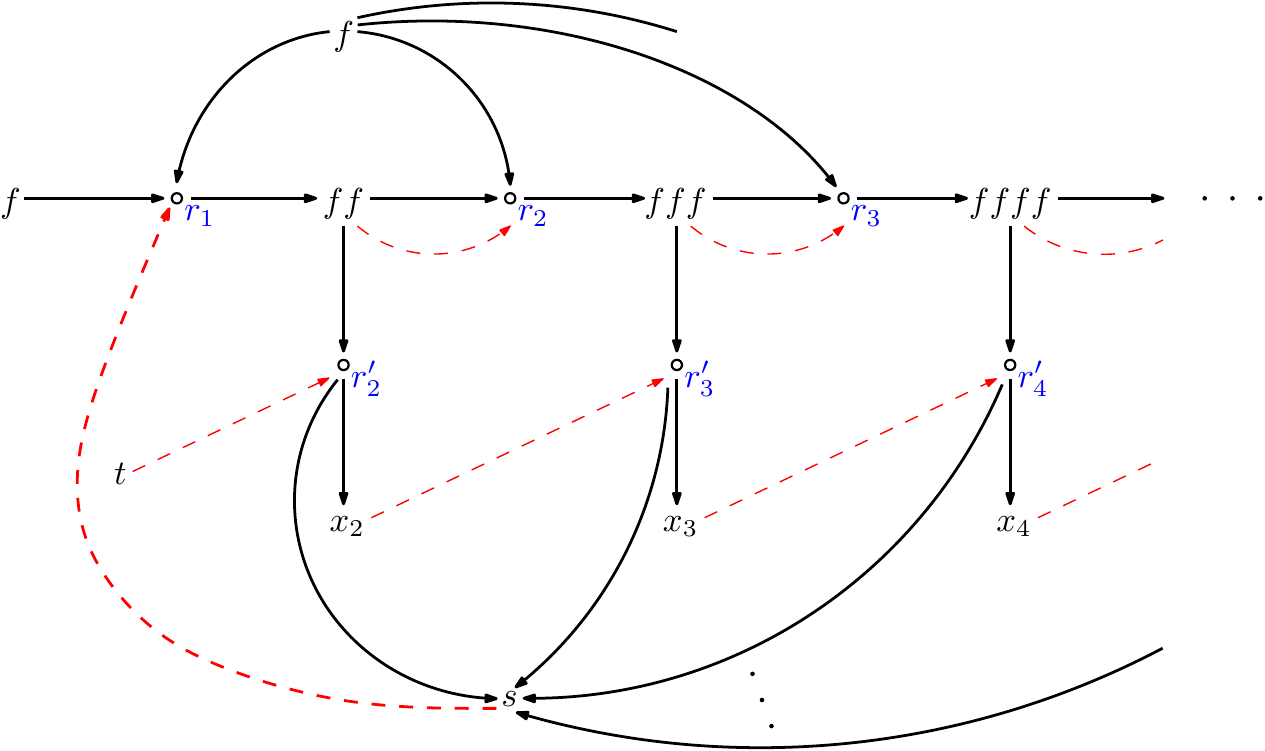}
\caption{An infinite CRS $\cQ_4$ which has no RAF even though $\mu(\cQ_4)$  is non-empty (equal to $\{r_1, r_2, \ldots\}$). This CRS satisfies (A1)--(A3) and (A5), but not (A4).}
\label{figy}
\end{figure}  
Notice that $\cQ_4$ satisfies (A1), (A2) and (A3). 
However, if we construct the sequence $\cR_i$ described above, then as the  sole  catalyst ($t$) of $r'_2$ is neither in the food set, nor generated by any other reaction,
it follows that $r'_2$ will be absent from $\cR_2$, and so $r'_3$ will also be absent from $\cR_3$ (since the only catalyst of $r'_3$ is produced by $r'_2$). Continuing in this way, we obtain $\mu(\cQ_4) = \{r_1, r_2, r_3, \ldots \}$, but this set is not an RAF, since the sole catalyst $s$ of $r_1$ does not lie lie in the closure of $F$ relative to $\{r_1, r_2, r_3, \ldots \}$ -- it was produced by the $r'_j$ reactions and in these have all disappeared in the limit; moreover it is clear that no subset of 
$\cQ_4$  is an RAF.
\hfill$\Box$

\bigskip

Thus, we require slightly stronger hypotheses than just (A1)--(A3)  in order to ensure that $\cQ$ has an RAF when $\mu(\cQ) \neq \emptyset$. This, is provided by the following result.

\begin{proposition}
\label{infp}
Let $\cQ =(X, \cR, C, F)$ satisfy (A1) and (A2). The following then hold:
\begin{itemize}
\item[(i)]  $\mu(\cQ)$ contains every RAF for $\cQ$; in particular, if $\mu(\cQ) = \emptyset$, then $\cQ$ has no RAF.
\item[(ii)]  Suppose that $\cQ$ satisfies both of the following further  conditions:
\begin{itemize}
\item[(A4)] $\bigcap_{i \geq 1} {\rm cl}_{\cR_i}(F) \subseteq {\rm cl}_{\mu(\cQ)}(F)$,   for the sequence $\cR_i$ defined in (\ref{R1eqx}).
\item[(A5)] Each reaction $r \in \cR$ is catalysed by only finitely many molecule types. 
\end{itemize}
\noindent Then $\cQ$ contains an RAF if and only if $\mu(\cQ)$ is non-empty (in which case, $\mu(\cQ)$ is the maximal RAF for $\cQ$).
\end{itemize}
\end{proposition}

Before proving this result, we pause to make some comments and observations concerning the new conditions (A4) and (A5). Regarding Condition (A4), containment in the opposite direction is automatic (by virtue of the fact that
$f(\cap Y_i)  \subseteq \cap_i f(Y_i)$ for any function $f$ and sets $Y_i$), so (A4) amounts to saying that the two sets described are equal.

Notice also that $\cQ_4$ in Example 2 (Fig.~\ref{figy}) satisfies (A5) but it violates (A4), as it must, since $\cQ_4$ does not have an RAF.  To see how $\cQ_4$ violates (A4), notice that
${\rm cl}_{\mu(\cQ_4)}(F) = \cF$, while $\bigcap_{i \geq 1} {\rm cl}_{\cR_i}(F) = \cF \cup \{s\}$.

Condition (A5) is quite strong, but Proposition~\ref{infp} is no longer true if it is removed. To see why, consider the following modification $\cQ'_4$ of $\cQ_4$ in which the only product of $r'_i$  (for $i>1$) is $x_i$, and $x_i$ catalyses $r_1$ for all $i>1$  (in addition to $r'_{i+1}$), as shown in Fig.~\ref{figz}.  Then  ${\rm cl}_{\mu(\cQ'_4)}(F) = \bigcap_{i \geq 1} {\rm cl}_{\cR_i}(F) = \cF$ so (A4) holds; however  $\mu(\cQ'_4) = \{r_1, r_2, \ldots\}$  which, as before, is not an RAF for $\cQ'_4$ since there is no catalyst of
$r_1$ in ${\rm cl}_{\{r_1, r_2, \ldots\}}(F)$.   Notice that (A5) fails for $\cQ'_4$ since $r_1$ has infinitely many catalysts.    Nevertheless, it is possible to obtain a result that dispenses with (A5)  at the expense of a  strengthening (A4), which we will do shortly in Proposition~\ref{infpro}.

\begin{figure}[htb]
\centering
\includegraphics[scale=0.8]{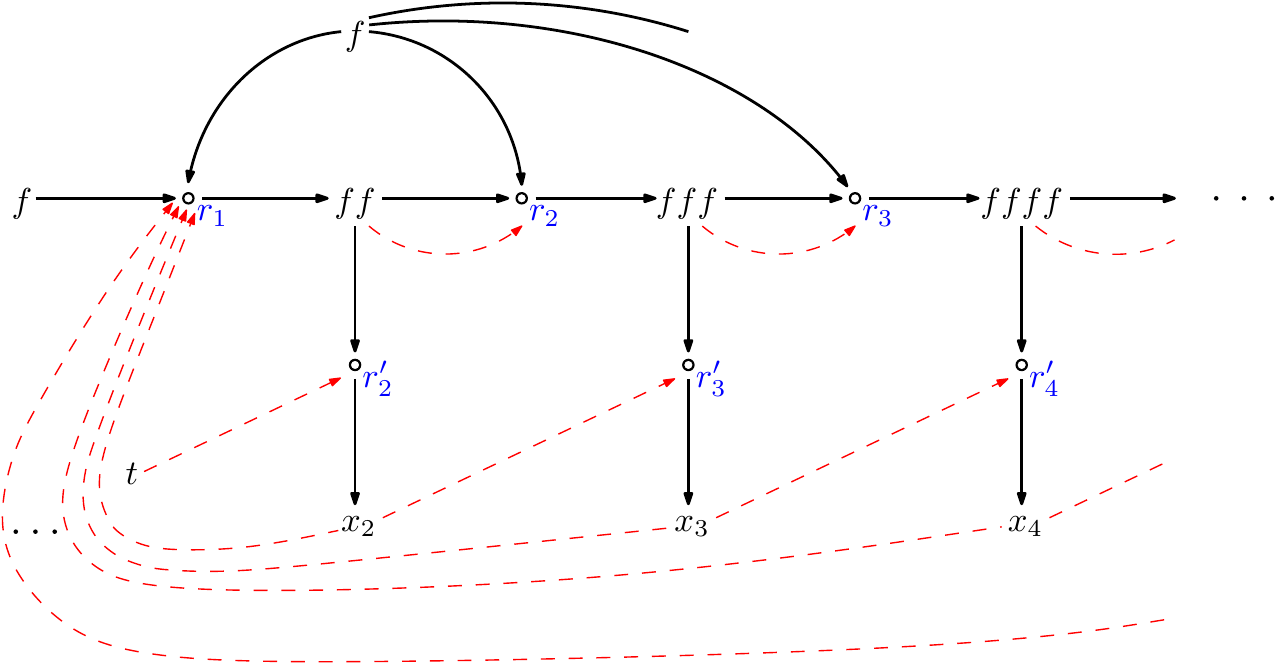}
\caption{An infinite CRS $\cQ'_4$ which has no RAF even though $\mu(\cQ'_4)$  is non-empty (equal to $\{r_1, r_2, \ldots\}$). This CRS satisfies (A1)--(A3) and (A4), but not (A5), nor (A4)$'$.}
\label{figz}
\end{figure}

\bigskip

{\em Proof of Proposition~\ref{infp}:}  Suppose $\cR'$ is any RAF for $\cQ$.  Induction on $i \geq 1$ shows that $\cR' \subseteq \cR_i$ for all $i$, so that $\cR' \subseteq \mu(\cQ)$; in particular, if $\mu(\cQ) = \emptyset$, then $\cQ$ has no RAF.  The proof of part (ii) of Proposition~\ref{infp} relies on a simple lemma.
\begin{lemma}
\label{simlem}
Suppose that $(A_i, i \geq 1)$  is any nested decreasing sequence of subsets and $B$ is a finite set for which $A_i \cap B \neq \emptyset$ for all $i\geq 1$. Then some element of $B$
is present in every set $A_i$. 
\end{lemma}
{\em Proof of lemma:} 
Suppose, to the contrary, that for every element $b \in B$,  there is some set $A_{i(b)}$  in the sequence that fails to contain $b$ (we will show this is not possible by deriving a contradiction).  Let $I  = \max\{i(b): b \in B\}$.  Since $B$ is a finite set, 
$I$ is a finite integer, and since the sequence $(A_i, i \geq 1)$ is a nested decreasing sequence, it follows that $A_I \cap B = \emptyset$, a contradiction. 
\hfill$\Box$

\bigskip

\noindent Returning to the proof of Part (ii), suppose that $\mu(\cQ) \neq \emptyset$; we will show that $\mu(\cQ)$ is an RAF for $\cQ$ (and so, by Part (i), the unique maximal RAF for $\cQ$). 
For $r \in \mu(\cQ), \rho(r) \subseteq cl_{\cR_i}(F)$ for each $i$ (otherwise $r$ would not be an element of $\cR_{i+1}$ and thereby fail to lie in $\mu(\cQ)$).  
Thus $\rho(r) \subseteq \bigcap_{i \geq 1} {\rm cl}_{\cR_i}(F) \subseteq {\rm cl}_{\mu(\cQ)}(F)$ by (A4).
It remains to show that $r$ is catalysed by at least one element of ${\rm cl}_{\mu(\cQ)}(F)$.  Let $B_r= \{x \in X: (x,r) \in C\}$.  By (A5), $B_r$ is finite. Moreover, for each $i \geq 1$,
$B_r \cap {\rm cl}_{\cR_i}(F) \neq \emptyset$ (otherwise $r$ would fail to be in $\cR_{i+1}$ and thereby not lie in $\mu(\cQ)$). By Lemma~\ref{simlem},
there is a molecule type $x \in B_r$ that lies in $\bigcap_{i \geq 1}{\rm cl}_{\cR_i}(F)$ and this latter set is contained in $ {\rm cl}_{\mu(\cQ)}(F)$ by (A4).  In summary, every reaction in 
$\mu(\cQ)$ has all its reactants and at least one catalyst present in ${\rm cl}_{\mu(\cQ)}(F)$ and so $\mu(\cQ)$ is an RAF for $\cQ$, as claimed.

\hfill$\Box$

Suppose we now remove Condition (A5) in Proposition~\ref{infp}. In this case, by a slight strengthening of (A4), we obtain a positive result (Proposition~\ref{infpro}). To describe this, we first require a further definition.  Recall that $C$ is the set of pairs $(x,r)$ where molecule type $x$ catalyses reaction $r$.  Given a subset $C'$ of $C$, let
$$\cR[C'] = \{r \in \cR: (x,r) \in C' \mbox{ for some } x\in X\}.$$
Define a nested decreasing sequence of subsets $C_1, C_2, \ldots, $ by $C_1 = C$ and for each $i\geq 1$,
\begin{equation}
\label{Cieqx} 
C_{i+1} = \{(x,r) \in C_i: \{x\} \cup \rho(r) \subseteq {\rm cl}_{\cR[C_i]}(F)\},
\end{equation}
and let $C_\infty = \bigcap_{i \geq 1} C_i$.

\begin{proposition}
\label{infpro}
Let $\cQ$ satisfy (A1) and (A2), as well as the following property:
$$(A4)'  \mbox{ } \bigcap_{i \geq 1} {\rm cl}_{\cR[C_i]}(F) \subseteq {\rm cl}_{R[C_\infty]}(F), \mbox{ for the sequence $C_i$ defined in (\ref{Cieqx})}.$$
Then $\cQ$ has an RAF if and only if $C_ \infty  \neq \emptyset$, in which case $\cR[C_\infty]$ is a maximal RAF for $\cQ$.

\end{proposition}
\begin{proof}
Suppose that $C_\infty \neq \emptyset$. Then for any $r \in \cR[C_\infty]$ there exists $x \in X$ such that $(x,r) \in C_\infty$.  It follows
that $(x,r) \in C_i$ for all $i$.  By definition, this means that $\{x\} \cup \rho(r) \subseteq {\rm cl}_{\cR[C_i]}(F)$ for all $i$, and so
$\{x\} \cup \rho(r) \subseteq \bigcap_{i \geq 1} {\rm cl}_{\cR[C_i]}(F).$  Now, by (A4)$'$, this means that $\{x\} \cup \rho(r) \subseteq {\rm cl}_{\cR[C_\infty]}(F)$.
In summary, every reaction in the non-empty set $\cR[C_\infty]$ has all its reactants and at least one catalyst in the closure of $F$ with respect to $\cR[C_\infty]$ and so
$\cR[C_\infty]$ forms an RAF for $\cQ$.  

Conversely, suppose that $\cQ$ contains an RAF $\cR'$; we will show that $C_\infty \neq \emptyset$.  For each $r \in \cR'$, select
a catalyst $x_r$ for $r$ for which $x_r \in {\rm cl}_{\cR'}(F)$.  Let $A= \{(x_r, r): r\in \cR'\}$. We use induction on $i$ to show that $A \subseteq C_i$ for all $i \geq 1$.  Clearly $A \subseteq C=C_1$, so suppose that $A \subseteq C_i$ and select an element
$(x_r, r) \in A$.  By definition, $$\{x\} \cup \rho(r) \subseteq {\rm cl}_{\cR'}(F) = {\rm cl}_{\cR[A]}(F) \subseteq  {\rm cl}_{\cR[C_i]}(F),$$
which means that $(x_r, r) \in C_{i+1}$, establishing the induction step. It follows that 
$\emptyset \neq A \subseteq \bigcap_{i\geq 1} C_i = C_\infty$ and so $C_\infty \neq \emptyset$ as claimed.

\hfill$\Box$ \end{proof}
Notice that, just as for condition (A4),  the condition (A4)$'$ is equivalent to requiring that the two sets described be identical.
Notice also that, although condition (A4) applies to the CRS $\cQ_4'$, condition (A4)$'$ fails, since
$C_\infty= \{(s, r_1), (ff, r_2), (fff, r_3), \ldots \}$ and so ${\rm cl}_{\cR[C_\infty]}(F) = \cF = \{f, ff, fff, \ldots\}$, while $s \in {\rm cl}_{\cR[C_i]}(F)$ for all $i\geq 1$,
and so $\bigcap_{i\geq 1} {\rm cl}_{\cR[C_i]}$ is not a subset of ${\rm cl}_{R[C_\infty]}(F)$.

In summary, a single application of $\mu$ allows us to determine when $\cQ$ has an RAF, provided the additional condition (A4)$'$ holds.  Example 2 showed that some additional assumption of this type is required, however one could also consider other approaches for determining the existence RAFs that do not assume a further condition like (A4)$'$, but instead iterate the map $\mu$.  In other words, consider the following `higher level' sequence of subsets of $\cR$:
$$\cR, \mu(\cQ), \mu^2(\cR), \cdots \mu^k(\cR) \cdots$$
where $\mu^k(\cR) = \mu(\mu^{k-1}(\cR)),$ for each $k\geq 2$.
Again, this forms a decreasing nested sequence of subsets of $\cR$ and so we can consider the set:
$$\nu(\cQ)= \bigcap_{i\geq 1} \mu^{i}(\cR).$$

In the example above for $\cQ_4$ where $\mu(\cQ) \neq \emptyset$, notice that $\mu^2(\cR) = \emptyset$ (and so $\nu(\cQ) = \emptyset$).
It follows from Proposition~\ref{infp} that if $\mu^k(\cQ)= \emptyset$ for any $k \geq 1$ then
$\cQ$ has no RAF.    However, just because $\nu(\cQ) \neq \emptyset$, this does not imply that $\cQ$ contains an RAF as the next example shows.

\bigskip

\noindent {\bf Example 3: }
Consider  the infinite CRS $\cQ_5= (X, \cR, C, F)$ which is obtained by taking a countably infinite number of (reaction and molecule disjoint) copies of $\cQ_4$ (from Example 2) and letting the molecule type $s$ in the $i$-th copy of $\cQ_4$ play the role of the molecule $t$ in the $(i+1)-$th  copy of $\cQ_4$.  In addition, let $r_0$ be the reaction $f+f+f\rightarrow \omega$ (where $\omega$ is an additional molecule)
catalysed by the $s$-products of all the copies of $\cQ_4$. Now $\mu^k(\cQ)$ contains all but the first $k$ copies of $\cQ_4$, plus $r_0$.  Consequently, $\nu(\cQ) = \{r_0\}$ but, as before, this is not an RAF. Notice, however that this example violates condition (A3).
\hfill$\Box$

\bigskip

\section{Finite RAFs in systems satisfying (A1)--(A3)}
\label{finitesec}

We have seen from the last section that applying $\mu$, even infinitely often,  does not seem to provide a way to determine whether a CRS possesses an RAF. However, in most applications, the main interest will generally be in finite RAFs. From the earlier theory it is clear that if  $\mu^k(\cQ)$ is finite for some integer $k \geq 1$ then any RAFs that may exist for $\cQ$ are necessarily finite, and finite in number. Moreover, if $$\emptyset \neq \mu^k(\cQ)=\mu^{(k+1)}(\cQ), \mbox{ for some } k\geq 1,$$
and this set is finite, then $\mu^k(\cQ)$ is  the unique (and necessarily finite) maximal RAF for $\cQ$. 
However, it is also quite possible that a CRS might contain both finite and infinite RAFs, and in this section we describe a characterisation of when an RAF contains a finite RAF.   

Given a CRS $\cQ$ define a sequence $\cR'_1, \cR'_2, \cdots$ of subsets of $\cR$ as follows:
$$\cR'_1 = \{r \in \cR: \rho(r) \subseteq F\}, \mbox{ and }$$ $$ \cR'_i = \{r \in \cR: \rho(r) \subseteq F \cup \bigcup_{1 \leq j<i}\pi(\cR'_j)\}, \mbox{  for each $i>1$}.$$
In words, $\cR_1$ is the set of reactions that have all their reactants in $F$, and for $i>1$ $\cR_i$ is the set of reactions for which each reactant is either an element of $F$ or products of some reaction in $\cR_j$ for $j<i$. 

\begin{proposition}
\label{finiteraf}
Suppose a CRS $\cQ$ satisfies (A1)--(A3).  
Let $\cQ_i' = (X, \cR'_i, C, F)$ for all $i \geq 1$, where $\cR'_i$ is as defined above. 
Then:
\begin{itemize}
\item[(i)] $(\cR'_i: i\geq 1)$ is a nested increasing sequence of finite sets.
\item[(ii)] $\cQ$ has a finite RAF if and only if $\mu(\cQ'_i)\neq \emptyset$ for some $i \geq 1$.
\item[(iii)] If $\mu(\cQ'_i) \neq \emptyset$ for some $i$,  then $\mu(\cQ'_j)$ is a finite RAF for $\cQ$ for all $j \geq i$.
\item[(iv)] Every finite RAF for $\cQ$ is contained in $\mu(\cQ_j)$ for some $j \geq 1$.
\end{itemize}

\end{proposition}
\begin{proof}

By (A1) and (A3),  it follows that $\cR'_1$ is finite, and, by induction, that $\cR'_i$ is finite for all $i>1$. Moreover, if $r \in \cR_i'$ then 
$\rho(r) \subseteq F \cup \bigcup_{1 \leq j<i}\pi(\cR'_j)$ and so $\rho(r) \subseteq F \cup \bigcup_{1 \leq j<i+1}\pi(\cR'_j)$ (i.e. $r \in \cR_{i+1}'$) and so the sets $\cR'_i, i\geq 1$ form an increasing nested sequence. This establishes (i).
For Parts (ii) and (iii), suppose that $\cQ$ contains a finite RAF $\cR'$. Since (A1) and (A2) hold, we can apply Lemma~\ref{lem1} to deduce that every reaction $r \in \cR'$ is an element of $\cR'_i$ for some $i$. Thus, since $\cR'$ is finite, and the sequence $\cR'_i$ is a nested increasing sequence of finite sets, it follows that $\cR' \subseteq \cR'_k$ for some fixed $k$, in which case $\mu(\cQ'_k) \neq \emptyset$.
Conversely, if $\mu(\cQ'_i) \neq \emptyset$, then it is clear from the definitions that $\mu(\cQ'_i)$ is an finite RAF for $\cQ$; moreover, so also is $\mu(\cQ'_j)$ for all $j>i$.   Part (iv) also follows easily from the definitions, since if $\cR'$ is a finite RAF for $\cQ$ then $\cR' \subseteq \cR'_j$ for some $j\geq 1$, and since $\cR'$ is finite we have $\mu(\cR')= \cR'$ and so
$\cR'=\mu(\cR') \subseteq \mu(\cQ_j)$.
This completes the proof.

\hfill$\Box$ \end{proof}
Theorem~\ref{finiteraf} provides an algorithm to search for finite RAFs in any infinite CRS that satisfies (A1)--(A3).    Given $\cQ$, construct $\cR'_1$ and run the (standard) RAF algorithm \cite{hor2} and \cite{hor3} on $\cR'_1$.  If it fails to find an RAF, then construct $\cR'_2$ and run the algorithm on this set, and continue in the same manner. If $\cQ$ contains a finite RAF, then this process is guaranteed to find it, however, there is no assurance in advance of how long this might take (if not constraint is placed on the size of the how large the smallest finite RAF might be).

\section{General setting}
Finally, we show how Proposition~\ref{infpro}  can be reformulated more abstractly in order to makes clear the
underlying mathematical principles; the added generality may also be useful for settings beyond chemical reaction systems.  This uses
the notion of ``$gf$-compatibility'' from \cite{hor1}, which we now explain. 

Suppose we have an arbitrary set $Y$ and  an arbitrary partially ordered set $W$, together with some functions $f:2^Y \rightarrow W \mbox{ and } g: Y \rightarrow W.$    Consider the function
$\psi: 2^Y \rightarrow  2^Y$, where $$\psi(A):=\{y \in A: g(y) \leq f(A)\}.$$
We are interested in the non-empty subsets of $Y$ fixed points of $\psi$, particularly, when   $f$ is {\em monotonic} (i.e., where $A \subseteq B \Rightarrow f(A) \leq f(B)$). 
A subset $A$ of $Y$ is said to be {\em $gf$-compatible} if $A$ is non-empty and $\psi(A) =A$.

The notion of an RAF can be captured in this general setting as follows. Given a CRS $\cQ= (X, \cR, F, C)$ satisfying (A2), take $Y=C$ and $W = 2^X$ (partially ordered by set inclusion), and define $f:2^Y \rightarrow W \mbox{ and } g: Y \rightarrow W$ as follows:
\begin{equation} 
\label{fgeq} 
f(A) = {\rm cl}_{\cR[A]}(F) \mbox{ and } g((x,r)) = \{x\} \cup \rho(r),
\end{equation} 
where, as earlier, $\cR[A]$ is the set of reactions $r \in \cR$ for which there is some $x' \in X$ with $(x',r) \in C$. 
Notice that $f$ is monotonic and when $\cQ$ is finite, the set $f(A)$ can be computed in polynomial time in the size of $\cQ$.
\begin{lemma}
\label{lemcom}
Suppose we have a CRS $\cQ$ satisfying (A2), and with $f$ and $g$ defined as in (\ref{fgeq}).
If $A$ is $gf$-compatible, then $\cR[A]$ is an RAF for $\cQ$.
Conversely, if $\cR'$  is an RAF for $\cQ$, then a $gf$-compatible set $A$ exists with
$\cR[A] = \cR'$.  In particular, $\cQ$ has an RAF if and only if $Y$ contains a $gf$-compatible set.
\end{lemma}
\begin{proof}
If $A$ is $gf$-compatible subset of $Y$, then for $\cR' = \cR[A]$, each reaction $r \in \cR$ has at least one molecule type $x \in X$ for which $(x,r) \in A$.  $gf$-compatibility ensures
that $g((x,r)) \subseteq f(A)$, in other words, $\{x\} \cup \rho(r) \subseteq {\rm cl}_{\cR'}(F) $ for some catalyst $x$ of $r$.  This holds for every $r \in \cR'$, so $\cR'$ is an RAF for $\cQ$.
Conversely, if $\cR'$ is an RAF, then for each reaction $r \in \cR'$, we can choose an associated catalyst $x_r$ so that $\{x_r\} \cup \rho(r) \subseteq {\rm cl}_{\cR'}(F) $. Then  $A=\{(x_r, r): r \in \cR'\}$
is a $gf$-compatible subset of $Y$, with $\cR[A] = \cR'$.
\hfill$\Box$ \end{proof}

The  problem of finding a $gf$-compatible set (if one exists) in a general setting (arbitrary $Y$, and $W$, not necessarily related to chemical reaction networks) can be solved in general polynomial time when $Y$ is finite and $f$ is monotonic and computable in finite time. This provides a natural generalization of the classical RAF algorithm.  In \cite{hor2}, we showed how other problems (including a toy problem in economics) could by  formulated within this more general framework.

However, if we allow the set $Y$ to be infinite, then monotonicity of $f$ needs to be supplemented with a further condition on $f$.
We will consider a  condition (`$\omega$-continuity'), which generalizes (A4)$'$, and  that applies automatically when $Y$ is finite.
We say that $f: 2^Y \rightarrow W$ is (weakly)  $\omega$-continuous if, for any nested descending chain $A_i, i \geq 1$ of sets, we have:
\begin{equation}
\label{glbeq}
f(\bigcap_{i \geq 1}A_i) \mbox{  is a greatest lower bound for }  \{f(A_i), i \geq 1\}.
\end{equation}
Recall that an element in a partially ordered set need not have a greatest lower bound (glb);  but if it does, it has a unique one. Notice that when $Y$ is finite, this property holds trivially, since then $f(\bigcap_{i \geq 1}A_i) = f(A_n)$ for the last set $A_n$ in the (finite) nested chain.

For a subset $A$ of $Y$ and $k\geq 1$, define $\psi^{(k)}(A)$  to be the result of applying function $\psi$ iteratively $k$ times starting with $A$.  Thus  $\psi^{(1)}(A)=\psi(A)$ and for $k\geq 1$, $\psi^{(k+1)}(A)=\psi(\psi^{(k)}(A))$.  Taking the particular interpretation of $f$ and $g$ in (\ref{fgeq}),  the sequence $\psi^{(k)}(Y)$ is nothing more than the sequence $C_k$ from (\ref{Cieqx}). 

Notice that the sequence $(\psi^{(k)}(A), k\geq 1)$ is a nested decreasing sequence of subsets of $Y$, and so we may define the set:
$$\overline{\psi}(A) :=\lim_{k\rightarrow \infty}\psi^{(k)}(A) = \bigcap_{k \geq 1} \psi^{(k)}(A),$$
which is a (possibly empty) subset of $Y$ (in the setting of Proposition~\ref{infpro}, $\overline{\psi}(A) = C_\infty$). 

Given (finite or infinite) sets $Y,W$, where $W$ is partially ordered, together with functions $f: 2^Y \rightarrow W \mbox{ and } g: Y \rightarrow W$, it is routine to  verify that  the following properties hold:
\begin{itemize}
\item[(i)] The $gf$-compatible subsets of $Y$ are precisely the non-empty subsets of $Y$ that are fixed points of $\psi$;
\item[(ii)] If $f$ is monotonic then $\overline{\psi}(Y)$ contains all $gf$-compatible subsets of $Y$; in particular, if  $\overline{\psi}(Y) = \emptyset$, then there is no $gf$--compatible subset of $Y$.
\item[(iii)] If  $f$ is  $\omega$-continuous then $\overline{\psi}(Y)$ is $gf$-compatible, provided it is non-empty; in particular, if $f$ is  monotonic and $\omega$-continuous then 
 (by (ii)) there  a $gf$-compatible subset of $Y$ exists if and only if $\overline{\psi}(Y)$ is nonempty.
\item[(iv)] Without the assumption that $f$ is weakly $\omega$-continuous in Part (iii),  it is possible for 
$\overline{\psi}(Y)$ to fail to be $gf$-compatible when $Y$ is infinite, even if $f$ is monotone.
\end{itemize}
The proof of Parts (i)--(iii) proceeds exactly as in \cite{hor1}, with the addition of one extra step required to justify Part (iii), assuming $\omega$-continuity. Namely, Condition (\ref{glbeq}) ensures that $\psi: 2^Y \rightarrow 2^Y$ is also $\omega$-continuous in the sense that for any nested descending chain $A_i, i \geq 1$ of sets, we have:
\begin{equation}
\label{psieq}
\psi(\bigcap_{i \geq 1}A_i) = \bigcap_{i \geq 1} \psi(A_i),
\end{equation}
and so $\psi(\overline{\psi}(Y)) =\overline{\psi}(Y)$. The proof of (\ref{psieq}) from (\ref{glbeq}) is straightforward: firstly, $\subseteq$ holds for {\em any} function $\psi$, while if $y \in \bigcap_{i \geq 1} \psi(A_i)$, then, by definition of $\psi$, $y \in A_i$ for all $i$ and $g(y) \leq f(A_i)$ for all $i\geq 1$ and so $y \in \bigcap_{i\geq 1} A_i$, and $g(y) \leq f(A_i)$ for all $i\geq 1$. Now, since $w=f(\bigcap_{i \geq 1}A_i)$ is a  glb  of $\{f(A_i): i \geq 1\}$, we have $g(y) \leq w$ for all $i$ (i.e.  $g(y) \leq f(\bigcap_{i \geq 1}A_i)$) and so $y \in \psi(\bigcap_{i \geq 1}A_i)$.  Part (iv) follows directly from Parts (ii) and (iii).

For Part (vi), consider the infinite CRS $\cQ_4$ in Example 2.    As above, take 
$Y= C,  W=2^X$ and, for $A \in 2^Y$, with $f$ and $g$ defined as in (\ref{fgeq}). 
Then  $\overline{\psi}(Y) = A$, where $A=\{(s, r_1), (ff, r_2), (fff, r_3), \ldots \}$ however, $A$ is not $gf$-compatible, since $(s, r_1) \in A$ and $g((s, r_1)) =\{s, f\}$ but this is not a subset of $f(A) = {\rm cl}_{\cR[A]}(F) = \cF$ since $s \not\in \cF$.
In this example, $f$ fails to be weakly $\omega$-continuous, and the argument is analogous to where we showed earlier that $\cQ'_4$ fails to satisfy (A4)$'$.  More precisely, for each $i\geq 1$, let $A_i =\{(c_r, r): r\in \cR_i\}$,   where $\cR_i$ is defined in (\ref{R1eqx}) and where, for each reaction $r \in \cQ_4$, $c_r$ is the unique catalyst of $r$.  Then  $f(A_i) = \cF \cup \{s\} \cup \{x_{i+1}, x_{i+2}, \ldots\}$ and so
$\bigcap_{i  \geq 1} f(A_i) = \cF \cup \{s\}$. 
However, $\bigcap_{i \geq 1} A_i =A$ 
  and so $f(\bigcap_{i \geq 1} A_i) = f(A) = \cF$,  which differs from the  glb of $\{f(A_i), i \geq 1\}$, namely $\bigcap_{i  \geq 1} f(A_i) = \cF \cup \{s\}$.
\hfill$\Box$

\section{Concluding comments}

The examples in this paper are particularly simple -- indeed mostly we took the food set to consist of just a single molecule, and reactions often had only one possible catalyst.  In reality more `realistic' examples can be constructed, based
on polymer models over an alphabet, however the details of those examples tends to obscure the underlying principles so we have kept with our somewhat `toy' examples in order that the reader can readily verify certain statements.

Section~\ref{finitesec} describes a process for determining whether an arbitrary infinite CRS (satisfying (A1)--(A3)) contains a finite  RAF. However, from an algorithmic point of view, Proposition~\ref{finiteraf}
is  somewhat limited, since the process described is not guaranteed to terminate in any given number of steps.  
If no further restriction is placed on the (infinite) CRS, then it would seem
difficult to hope for any sort of meaningful algorithm; however, if the CRS has a `finite description' (as do our main examples above),  then the question of the algorithmic decidability of the existence of an RAF or of a finite RAF arises.  
More precisely, suppose an infinite CRS $\cQ = (X, \cR, C, F)$ consists of (i) a countable set of molecule types $X= \{x_1, x_2, \ldots \}$, where we may assume (in line with  (A1)) that
$F = \{x_i: i<K\}$, for some finite value $K$, and (ii) a countable set 
$\cR = \{r_1, r_2, \ldots\}$ of reactions, 
where 
$r_i$ has a finite set $\alpha(i)$ of reactants, a finite set  $\beta(i)$ of products, and a finite or countable set $\gamma(i)$ of catalysts,  where  $\alpha, \beta$ and $ \gamma$ are  computable (i.e. partial recursive) set-valued functions defined on the positive integers. 
Given this setting, a possible question for further investigation is whether (and under what conditions) there exists an algorithm to determine whether or not $\cQ$ contains an RAF, or more specifically a finite RAF  (i.e. when is this question decidable?).

\section{Acknowledgements}
The author thanks the Allan Wilson Centre for funding support, and Wim Hordijk for some useful comments on an earlier version of this manuscript.   I also thank Marco Stenico (personal communication) for pointing out that $\omega$-consistency is required for Part (iii) of the  $gf$-compatibility result above  when $Y$ is infinite, and for a reference to a related fixed-point result in domain theory (Theorem 2.3 in \cite{Stoltberg:94}), from which this result can also be derived.


\begin{thebibliography}{}
%
%
\bibitem{dit} P. Dittrich,  P. Speroni di Fenizio, Chemical organisation theory.  Bull. Math. Biol. \textbf{69}, 1199--1231 (2007)
\bibitem{hig} P. G. Higgs,  N. Lehman,  The RNA World: molecular cooperation at the origins of life. Nat. Rev. Genet. \textbf{16}(1), 7--17 (2015)
\bibitem{hor1} W. Hordijk, M. Steel, Autocatalytic sets extended: dynamics, inhibition, and a generalization. J. Syst. Chem. \textbf{3}:5 (2012)
\bibitem{horn} W. Hordijk, M. Steel,  Autocatalytic sets and boundaries.  J. Syst. Chem. (in press) (2014)
\bibitem{hor2} W. Hordijk,  M. Steel, Detecting autocatalyctic, self-sustaining sets in chemical reaction systems. J. Theor. Biol. \textbf{227}(4), 451--461 (2004)
\bibitem{hor3} W. Hordijk, S. Kauffman, M. Steel, Required levels of catalysis for the emergence of autocatalytic sets in models of chemical reaction systems. Int. J.  Mol. Sci. (Special issue: Origin of Life 2011) \textbf{12}, 3085--3101 (2011)
\bibitem{hor4} W. Hordijk, M. Steel, S. Kauffman,  The structure of autocatalytic sets:   evolvability, enablement, and emergence. Acta Biotheor. \textbf{60}, 379--392 (2012)
\bibitem{kau1} S. A. Kauffman,   Autocatalytic sets of proteins. J. Theor. Biol. \textbf{119}, 1--24 (1986)
\bibitem{kau2} S. A. Kauffman,   The Origins of Order (Oxford University Press, Oxford 1993)
\bibitem{mos}  E. Mossel,  M. Steel,   Random biochemical networks and the probability of self-sustaining autocatalysis. J. Theor. Biol. \textbf{233}(3), 327--336 (2005)
\bibitem{sm} J. Smith, M. Steel,  W. Hordijk,   Autocatalytic sets in a partitioned biochemical network. J. Syst. Chem. \textbf{5}:2 (2014)
\bibitem{ste}  M. Steel, W. Hordijk, J. Smith,  Minimal autocatalytic networks. Journal of Theoretical Biology {\bf 332}: 96--107 (2013)
\bibitem{Stoltberg:94}  V. Stoltenberg-Hansen, I.  Lindstr{\"o}m, E.  R. Griffor,  Mathematical Theory of Domains. Cambridge Tracts in Theoretical Computer Science 22 (Cambridge University Press, Cambridge 1994)
\bibitem{vas} V. Vasas,  C. Fernando, M. Santos, S. Kauffman, E. Szathm{\'a}ry,   Evolution before genes. Biol. Dir. \textbf{7}:1 (2012)
\bibitem{vil} M. Villani, A. Filisetti, A. Graudenzi , C. Damiani, T. Carletti, R. Serra,  Growth and division in a dynamic protocell model.  Life  \textbf{4}, 837--864 (2014)
\end{thebibliography}
\end{document}